\documentclass[conference]{IEEEtran}
\IEEEoverridecommandlockouts
\usepackage{cite}
\usepackage{amsmath,amssymb,amsfonts,amsthm}
\usepackage{algorithmic}
\usepackage{graphicx}
\usepackage{textcomp}
\usepackage{xcolor}

\newenvironment{environment}

\usepackage{balance} 

\usepackage{soul} 
\usepackage{times}
\usepackage{hyperref}
\usepackage{microtype}
\usepackage{color}
\usepackage{cleveref}
\usepackage{comment}
\usepackage[ruled,vlined]{algorithm2e}
\usepackage[colorinlistoftodos]{todonotes}
\usepackage{xcolor}
\usepackage{thmtools}
\usepackage{thm-restate}

\usepackage{algorithm2e}
\SetKwInput{KwData}{Input}
\newcommand{\maxsize}{\textsc{M}}

\newcommand{\utility}{\mathcal{U}_{u_S}}
\newcommand{\revenue}{\mathbb{E}^{rev}_u}
\newcommand{\fees}{\mathbb{E}^{fees}_u}
\newcommand{\transacts}{p^{trans}_{u,v}}

\newcommand{\channelscost}{\sum_{(v,l) \in S} L_u(v,l)}

\newcommand{\subsectioninline}[1]{\noindent \textbf{#1:}}

\usepackage{graphicx}
\graphicspath{ {./img/} }

\newtheorem{theorem}{Theorem}

\def\BibTeX{{\rm B\kern-.05em{\sc i\kern-.025em b}\kern-.08em
    T\kern-.1667em\lower.7ex\hbox{E}\kern-.125emX}}
\begin{document}

\title{Lightning Creation Games}

\makeatletter
\newcommand{\linebreakand}{%
  \end{@IEEEauthorhalign}
  \hfill\mbox{}\par
  \mbox{}\hfill\begin{@IEEEauthorhalign}
}
\makeatother

\author{\IEEEauthorblockN{Zeta  Avarikioti}
\IEEEauthorblockA{
\textit{TU Wien}\\
Vienna, Austria \\
georgia.avarikioti@tuwien.ac.at}
\and
\IEEEauthorblockN{Tomasz Lizurej}
\IEEEauthorblockA{
\textit{University of Warsaw $\&$ IDEAS NCBR}\\
Warsaw, Poland \\
tomasz.lizurej@crypto.edu.pl}
\and
\IEEEauthorblockN{Tomasz Michalak}
\IEEEauthorblockA{
\textit{University of Warsaw $\&$ IDEAS NCBR}\\
Warsaw, Poland \\
tpm@mimuw.edu.pl}
\linebreakand 
\IEEEauthorblockN{Michelle Yeo}
\IEEEauthorblockA{
\textit{Institute of Science and Technology Austria}\\
Klosterneuburg, Austria \\
myeo@ist.ac.at}}

\maketitle

\begin{abstract}
Payment channel networks (PCNs) are a promising solution to the scalability problem of cryptocurrencies. 
Any two users connected by a payment channel in the network can theoretically send an unbounded number of instant, costless transactions between them.
Users who are not directly connected can also transact with each other in a multi-hop fashion.
In this work, we study the incentive structure behind the creation of payment channel networks, particularly from the point of view of a single user that wants to join the network. 
We define a utility function for a new user in terms of expected revenue, expected fees, and the cost of creating channels, and then provide constant factor approximation algorithms that optimise the utility function given a certain budget.
Additionally, we take a step back from a single user to the whole network and examine the parameter spaces under which simple graph topologies form a Nash equilibrium. 
\end{abstract}


\section{Introduction}
\noindent One of the critical limitations of the major cryptocurrencies, such as Bitcoin or Ethereum, is their low transaction throughput~\cite{chauhan2018blockchain,croman2016scaling,jain2021we}. For instance, given Bitcoin's block size limit of 1MB and the average block creation time of 10 minutes, its throughput  is limited to tens of transactions per second. This is clearly not enough to facilitate the widespread everyday use of Bitcoin. For comparison, credit card payment systems such as VISA handles approximately
7K transactions per second~\cite{visa}.
 
Payment Channel Networks (PCNs), such as Bitcoin's Lightning Network~\cite{lightning} and Ethereum's Raiden Network~\cite{utomo2020blockchain}, are second-layer solutions that are designed to address the above scalability problem. The core idea is to process the majority of transactions off-chain by enabling nodes to establish bilateral payment channels; each channel acts as a joint account between the channel participants. To preserve security, opening a channel requires depositing funds to a shared address on-chain. These funds serve as secure collateral to possibly many off-chain transactions between both parties. When the channel is closed, the final balance is settled on-chain.

Importantly, each node can establish such payment channels with many other nodes. This gives rise to a network that allows for funds transfers to non-neighbors through a path of intermediaries. 
Because opening and maintaining a channel requires locking up funds, serving as an intermediary results in opportunity costs.  
To mitigate this cost, intermediary nodes earn transaction fees for their services.

The protocols underlying PCNs have attracted a lot of attention in the literature~\cite{gudgeon2020sok}. In addition to analyzing cryptographic underpinnings of the PCN’s security proofs~\cite{kiayias2019composable,malavolta2017concurrency}, an effort has been made to understand game-theoretic aspects of these networks either with respect to security e.g.,~\cite{rain2021towards,avarikioti2021brick,avarikioti2022suborn,mccorry2019pisa,avarikioti2020cerberus}, or economics, e.g.,~\cite{wangriding,engelshoven2021merchant}. 

A particularly interesting question is how the nodes should choose \emph{where   to connect} to a PCN and \emph{what portion of a budget should be locked} to distinct channels. This is important as their choice not only affects the situation of individual nodes but also influences the resulting network as a whole. However, this issue has been weakly studied in the literature. In  fact, most PCN implementations (e.g., the Lightning Network) still propose a simple heuristic for new nodes, suggesting connecting to a trusted peer or a hub.

In this work, we answer this question by first presenting several attachment strategies for newly-joining nodes in a PCN.
The first key challenge to this task is to define the new node's utility function that accurately reflects the key objectives of new PCN users. A newcomer has to weigh the cost of creating channels and locking up capital against the profits stemming from these connections and the node's position in the network. Furthermore,  the utility function should be efficiently computable, so that it can be used in practice by new nodes, posing a second challenge.

Unfortunately, the models of the utility function considered so far in the literature do not take all the above aspects into account. In particular, Guasoni et al.~\cite{lnecon} analyse the cost of channel creation, and establish conditions under which two parties would create unidirectional or bidirectional channels between themselves, as opposed to transacting on-chain. However, the utility function in \cite{lnecon} only accounts for the cost of channel creation but neglects profits from routing transactions and fees a user could encounter. Avarikioti et.\ al~\cite{avarikioti2019payment,avarikioti2020ride} and Ersoy et.\ al~\cite{ersoy2019profit}, on the other hand, account for fees and profits from routing transactions through the PCN but neglect the opportunity costs from the locked capital and consider only a simplified transaction model where users transact with each other with uniform probability.

We take up the first challenge to define a utility function that accurately depicts the gains and costs of newly joining nodes. In particular, we account for on-chain costs for opening channels, routing fees paid to and by the node due to its position in the network, and opportunity costs for locked capital. We further leverage a realistic transaction distribution where nodes transact with other nodes with probability proportional to their degree, inspired by the well-known Barab{\'a}si-Albert preferential attachment model~\cite{barabasi1999emergence}. We believe this transaction distribution approximates well real-life scenarios where nodes transact more often with big vendors and service providers.
We further address the second challenge by providing a series of approximation algorithms to efficiently compute the optimal connection strategy for newly-joining nodes.
The approximation ratio and runtime of each algorithm depend on how much freedom the node has to distribute its budget on the channels, highlighting an interesting trade-off.

Apart from the myopic analysis for a single joining node, we also examine the effect our strategies may have on the topological structure of a PCN. In particular, we examine simple graph structures, i.e., path, circle, and star graphs, to determine under which conditions these constitute stable graphs, where no node may increase its utility by changing its strategy (Nash equilibrium).
Naturally, which topologies are stable or not heavily depends on the parameters of the transaction distribution. We thus identify the exact parameter space in which each topology constitutes a Nash equilibrium.

In summary, our  contributions are as follows:

\begin{itemize}
    \item We extend the utility function of~\cite{avarikioti2020ride} to incorporate a \emph{realistic transaction model and opportunity costs}. To that end, we consider transaction distributions where users transact with other users in proportion to their degree instead of uniformly at random as in~\cite{avarikioti2020ride,ersoy2019profit,avarikioti2019payment}. 

    \item We provide a series of \emph{approximation algorithms} that maximize our utility function under different constraints. In particular, we identify a \emph{trade-off between the runtime of the algorithm and capital distribution constraints}, i.e., how much capital is locked in each channel at its creation.
    
    
    \item Finally, we examine simple graph topologies and determine under which parameter space of the transaction distribution, they form \emph{Nash equilibria}.
    
\end{itemize}


\section{The Model}
\label{sec:preliminaries}
In this section, we outline our model which is an extension of the model introduced in~\cite{avarikioti2020ride}. 
We alleviate several unrealistic assumptions introduced in \cite{avarikioti2020ride}, thus providing more meaningful insights on the connection strategies and expected network structure of PCNs. We indicate these assumptions below.

\subsection{Payment channel networks and routing fees} 
Payment channels provide a way for users on the blockchain to transact directly with each other off-chain, thereby avoiding the high fees and latency involved in transacting on the blockchain. 
Any two users on the blockchain can open a payment channel with each other by locking some of their funds to be used only in this channel, much like opening a joint account in a bank. 
Once the channel is created, both users can send each other coins by updating the channel balances in favour of the other party (see~\Cref{fig:channel_example} for an example). 
For each payment (channel balance update), the respective capital must be respected, meaning that a party cannot send more coins than it currently owns to the counterparty.
To close their payment channel, the parties post on-chain a transaction that depicts the latest mutually agreed distribution of their funds. The closing transaction can be posted either in collaboration or unilaterally by one channel party. Note that \emph{posting a transaction on-chain bears a cost}: the fee to the miner that includes the transaction on the blockchain.

A payment channel network comprises of several two-party channels among users of the blockchain. 
Each user of the network is represented by a vertex while each (bidirectional) channel among two parties is represented by $2$ directed edges (one in each direction) connecting the two vertices corresponding to the parties. 
We model each bidirectional channel as $2$ directed edges to take into account the balance on both ends of the channel which can be different and thus impose different limits on the payment amount that can be sent in each direction.
More concretely, let us represent the topology of a payment channel network with a directed graph $G=(V,E)$ with $|V| = n$, and $|E| = m$.
For node $u \in V$, let $Ne(u)$ denote the set of in- and out-neighbors of $u$.

\begin{figure}[htb!]
    \centering
    \includegraphics[width=0.65\linewidth]{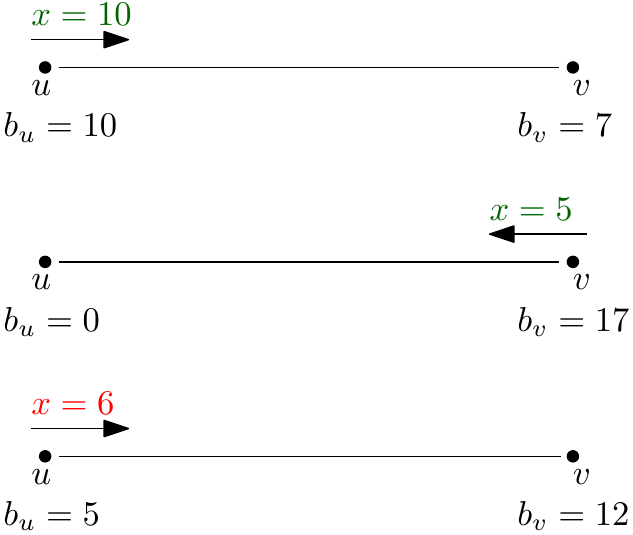}
    \caption{Example of payments going through a channel between $2$ users $u$ and $v$. $b_u$ and $b_v$ denote the balances of $u$ and $v$ in the channel and are updated with every successful payment. The last payment of size $6$ going from $u$ to $v$ is unsuccessful as the size of the payment is larger than $b_u=5$.}
    \label{fig:channel_example}
\end{figure}

Users who are not directly connected with a channel can still transact with each other if there exists a path of channels between them in the PCN graph. For instance, if Alice and Carol do not share a channel, but Alice shares a channel with Bob and Bob with Carol then Alice may send coins to Bob and then Bob to Carol\footnote{There exist techniques, namely HTLCs, to ensure that the transactions on a path will be executed atomically, either all or none, so the intermediaries do not lose any funds~\cite{lightning}.}.
However, each channel must hold enough coins to ensure the feasibility of the transaction routing. In our previous example, if Alice wants to send to Carol 5 coins through Bob, then Alice must own at least 5 coins in her channel with Bob, and Bob at least 5 coins in his channel with Carol, at the time of the transaction execution. 
The users along the transaction path who are not the sender or receiver (e.g., Bob) typically \emph{charge a fee for forwarding the transaction} that depends on the transaction amount and is publicly announced. The precise form of the fee function for forwarding transactions is determined by the user who owns the coins. That is, given a payment channel $(u,v)$ and a fixed transaction amount of $t$, the fees incurred from forwarding $t$ can even differ depending on whether $t$ is forwarded from $u$ to $v$ or from $v$ to $u$.

In our model, we assume transactions ($tx$) are of size at most $T >0$ and all intermediary nodes take the same -- global -- fee function $F: [0,T] \longrightarrow \mathbb{R}^+$ which is an abstraction for an average fee function. We denote by $f_{avg}$ the value of the average fee when using the global fee function $F$. 
That is, $f_{avg} = \int_0^T p_{tx \ \text{size} = t} \cdot F(t) dt$, where $p_{\text{tx size} = t}$ is a global probability of occurrence of a transaction with size $t$. We assume that $f_{avg}$ is publicly known (recall that the fee functions are publicly announced in PCNs).

\subsection{PCN transactions}\label{subsec:transactions}
In the following, we alleviate the assumption of \cite{avarikioti2020ride} that transactions are uniformly distributed among the PCN users, and introduce a more realistic transaction model.

\subsectioninline{Transactions}
\label{subsection:model}
Let $N_u$ denote the average number of transactions sent from user $u$ over a unit of time. We denote with $N$ the sum of the number of all transactions sent by users in a unit of time $N = \sum_{v \in V} N_v$.
We assume a user $u$ joining the network knows the distribution of transactions in the network.
 These assumptions equally allow each user to estimate the mean rate (denoted by $\lambda_{uv}$) of transactions going along any directed edge $(u,v)$ which, we assume, follows a Poisson process with rate $\lambda_{uv}$. 
We also stress that this estimation can be done efficiently in time $\mathcal{O}(n^2)$, by calculating shortest paths using e.g., Dijkstra's algorithm~\cite{Dijkstra59} for each pair of nodes in the network.

\subsectioninline{Reduced subgraph with updated capacities:}
The topology of the PCN can change with the size of transactions due to balance constraints: some directed edges do not have enough capacity to forward transactions when the transaction size is too large.  
However, given that we assume users know the distribution of transactions in the network, and that the capacity and time of channel creation are publicly posted on the blockchain, users can estimate the expected balance on each end of all channels in the network. 
Thus, for the rest of the paper, we consider that all our proposed algorithms for a given transaction of size $x$ are computed on a subgraph $G'$ of the original PCN $G$ that only takes into account directed edges that have enough capacity to forward $x$.


\subsectioninline{Transaction distribution}\label{subsection:model_params}
In the topological studies included in this work, we assume that the probability that any two users transact with each other is proportionate to their degree. Specifically, we use the \textit{Zipf distribution}~\cite{zipf} to model the occurrence of any two users transacting with each other. That is, assume for a user $u$ a ranking of all other users in the network according to their degree, breaking ties arbitrarily. That is, the highest degree vertex is given rank 1, the second highest is given rank 2, etc. Then for some user-specific parameter $s_u > 0$, the probability $p^{trans}_{u,v}$ that $u$ transacts with another user $v \in V\setminus \{u\}$ with rank $k$ is:

\begin{equation}\label{eq:prob}
p^{trans}_{u,v} = \frac{1/k^{s_u}}{\sum_{i=1}^n 1/i^{s_u}}.
\end{equation}

We note that the Zipf distribution is widely used in the natural and social sciences for modelling data with a power law distribution \cite{Salge2015,Aitchison2016}. It is also frequently used in the context of social networks~\cite{Bilo0LLM21} and thus seems a natural model for approximating the probability of any 2 users transacting in a payment channel network.

Let the edge betweenness centrality be defined as:

$$EBC(e) := \sum_{s, r \in V; s \neq r; m(s,r)>0} \frac{m_e(s,r)}{m(s,r)},$$
\noindent where $m_e(s,r)$ is the number of shortest paths that traverse through the edge $e$ and $m(s,r)$ is the total number of shortest paths from $s$ to $r$. The transaction rate $\lambda_e$ for all directed edges $e$ in $E$ can be estimated by the edge betweenness centrality of the edge $e$ weighted by the probability of any two vertices $s$ and $r$ transacting with each other. That is, for a directed edge $e$, we first define the probability $p_e$ that the edge $e$ is chosen in a single transaction:

\begin{equation}\label{eq:pe}
    p_e = \sum_{s, r \in V; s \neq r; m(s,r)>0} \frac{m_e(s,r)}{m(s,r)} p^{trans}_{s,r}.
\end{equation}

Let $N$ denote the average number of transactions that happen in a unit of time sent out by a user in the network. 
We assume these transactions are independent. 
The average number of times a directed edge $e=(u,v)$ is chosen in $N$ transactions is the transaction rate $\lambda_{e}$ and is simply $N\cdot p_{e}$.

In this work, we slightly modify the original Zipf distribution to ensure that the probability of any user transacting with two other distinct users having the same degree is equal. 
We do this by simply averaging the Zipf probability of transacting with every user with the same degree.
Below we propose a detailed method of calculating the probability that a given node $u$ transacts with any other node $v$ in the network.

Given a network $G = (V, E)$,  we first consider the subgraph $G' = (V' = V \setminus \{u\}, E')$ which is created by removing the node $u$ all of its incident edges from $G$.
Then, we sort all nodes in $V'$ by their \emph{in-degree} and then assign a \emph{rank-factor} -- $rf(v)$ to each node $v$ in $V'$. 
Since we want to ensure that every node with the same in-degree has the same rank-factor, we simply average the ranks of nodes with the same in-degree. 
In more detail, let $r_0(v) $ denote the smallest rank of a node $v' \in V'$ such that the in-degree of $v'$ is equal to the in-degree of $v$. 
Let $n(v)$ be the number of nodes in $V'$ with the same in-degree as $v$. 
The rank factor of $v$ can be computed as follows: 
$$rf(v) = \frac{\frac{1}{r_0^s(v)} + \ldots + \frac{1}{(r_0(v)+n(v))^s}}{n(v)}$$
The probability that $u$ transacts with $v \in V'$ is then:
$$\transacts= \frac{rf(v)}{\sum_{v' \in V'} rf(v')}$$

Finally, observe that the modified Zipf distribution satisfies the following property: $r_1(v_1) < r_2(v_2) \implies rf(v_1) > rf(v_2)$. This holds because $rf(v_1) \geq \frac{1}{(r_0(v_1)+n(v_1))^s}$ and $rf(v_2) \leq \frac{1}{(r_0(v_2))^s}$.

\subsection{Utility function of a new user}\label{sec:utility}
When a new user joins a PCN, they must decide \emph{which channels to create and how much capital to lock in each channel}, while respecting their own budget. In their decision, the user must factor the following: (a) the on-chain costs of the channels they choose to open, (b) the opportunity cost from locking their capital for the lifetime of each channel, (c)  the potential gains from routing transactions of others (routing fees),  (d) the routing fees they must pay to route their own transactions through the PCN, (e) their budget.
Intuitively, the more channels a user opens and the higher the amount of the total capital locked, the more fees they will obtain from routing and the less cost they will bear for routing their own transactions. In other words, increasing the initial costs also increases the potential gains. 
Our goal is to analyze these trade-offs and find the sweet spot that maximizes the benefits for a newly-joining user with a specific budget. 
We account for all these factors in a realistic manner when we design the utility function of the user, in contrast to previous work~ \cite{avarikioti2020ride,avarikioti2019payment} where the opportunity cost was omitted, and the routing fees were calculated naively (i.e., constant fees and uniform transaction distribution). 

\begin{figure}[t]
    \centering
    \includegraphics[width=0.6\linewidth]{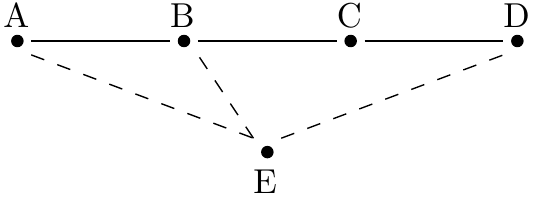}
    \caption{$E$ joins a PCN with existing users $A$, $B$, $C$, $D$.
    $E$ plans to transact with $B$ once a month, and $A$ usually makes $9$ transactions with $D$ each month.
    We assume the transactions are of equal size,  and transaction fees and costs are of equal size.
    $E$ has enough budget only for $2$ channels, with the spare amount of funds to lock equaling $19$ coins. $E$ should create channels with $A$ and $D$ of sizes $10$ and $9$ to maximize the intermediary revenue and minimize $E$'s own transaction costs.
    }
    \label{fig:example}
\end{figure}

\subsectioninline{User strategy and constraints}
Consider a fixed PCN $G = (V,E)$ and a new user $u$ that wants to join $G$.
Furthermore, let us denote the set of possible actions by $\Omega := \{(v_i, l_i)\}_{i}$, where each element $(v_i, l_i) \in \Omega$ represents a node $v_i$ that $u$ wants to connect to by locking in an amount of $l_i >0$ on the corresponding channel. The strategy of $u$ is to select a set of users (a strategy) $S \subset \Omega$ that $u$ wants to connect to and how much funds to deposit in these channels. Note that both $\Omega$ and $S$ may contain more than one channel with the same endpoints, but different amounts of locked funds on each end. 
We also assume $u$ has some budget $B_u > 0$ to create and fund channels and $u$'s budget constraint imposes the requirement that for the strategy $S \subseteq \Omega$ chosen by $u$, $\sum_{j=1}^{|S|} [C + l_j ] \leq B_u$. Finally, we remark that $\Omega := \{(v_i, l_i)\}_{i}$ may contain $v_i$ with continuously many values of $0 \leq l_i \leq B_u$. We will call it a continuous action set. In this case, we will operate on the set of vertices $\Omega^V \subseteq V$ for which the user $u$ will choose a strategy $S$ consisting of pairs $(x_j, l_j): x \in \Omega^V, 0 \leq l_j \leq B_u$.
\Cref{fig:example} highlights a simple example of the decision-making process of a new user that wants to join an existing PCN.

Now we define the \emph{utility function} for a new user $u$ that wants to join a PCN $G$ and has a fixed budget $B_u$. 
The goal of $u$ is to \emph{choose any strategy $S = \{(v_i, l_i)\}_{i} \subseteq \Omega $ to maximize their expected profit within a given budget $B_u$.}  
The expected profit (utility) of $u$ is essentially the expected revenue from forwarding transactions through its channels and collecting the routing fees, minus the costs of creating the channels (on-chain fees and opportunity cost) and the expected fees encountered by $u$ when sending transactions to other users in the network.

\subsectioninline{Channel costs} \label{sec:costs}
Typically, two on-chain transactions are needed to open and close a channel~\footnote{We omit the case when one of the parties may commit fraud, and a third transaction is necessary to award the cheated party all the channel funds. This case is outside the scope of the paper as the costs for publishing such a transaction may be covered by the total funds of the channel}.
Recall that each blockchain transaction costs a fee to the miners, denoted $C$. 

The cost of the \emph{opening} transaction can be shared by two parties, and we assume that parties only agree to open channels if they share this cost equally ($C/2$ each).
The cost of the closing transaction, on the other hand, is paid by both parties when the channel closes in collaboration, or by the party that closes the channel when the channel is closed unilaterally. To model the costs of a channel between $u$ and $v$, we  assume that it is equally probable that the channel closes in one of the three ways: unilaterally by $v$, unilaterally by $u$, in collaboration of $u$ and $v$. Thus, the cost of the closing transaction is on expectation $C/2$ for each party. Hence, in total, the channel cost for each party is $C$.

We also account for the opportunity cost of locking funds (as opposed to using or storing them elsewhere) in a channel for the lifetime of the channel.
Suppose two users $u$ and $v$ wish to open a channel locking $c_u$ and $c_v$ amount of coins respectively. Let $l_i, i \in V$ the opportunity cost defined by user $i$; that is, typically a function of the amount of coins $c_i$, e.g., $l_i = r \cdot c_i$, $r$ constant (a standard economic assumption due to the non-specialized nature of the underlying coins~\cite{hall1998macroeconomics}).
We denote the total cost for opening a channel for user $u$ by  $L_u(v,l) = C + l_u$. The cost of user $v$ is symmetric.

We direct the reader to the work by Guasoni et al.~\cite{lnecon} for a more detailed model of channel costs. 
We note that our computational results still hold in this extended model of channel cost. We further note that while the utility function in \cite{lnecon} only accounts for the cost of channel creation, in our work we also consider  the potential profit from routing transactions and fees a user could encounter.

    







\subsectioninline{Revenue from routing fees}
%
Each user of the PCN may route transactions of others through their channels in exchange for a routing fee. Each time a user $u$ provides such service, $u$ gains revenue equal to $f_{avg}$ as described in section~\ref{subsec:transactions}.
Specifically, the expected revenue gained by a user $u$ over a unit time interval from routing transactions in the PCN is the sum of the fees weighted by the average transaction rate from all of $u$'s incident channels: 




\begin{equation}\label{eq:rev}
    \revenue = \sum_{v_i \in Ne(u)} \lambda_{ u v_i} \cdot f_{avg}.
\end{equation}
We write $\mathbb{E}^{rev}_{u_S}$ when we want to explicitly say that the user $u$ already added edges from $S$ to the network.

\subsectioninline{Fees encountered by the user}
Whenever a user in the network $u$ makes a payment to another user $v$, $u$ has to pay some amount of fees to all the intermediary nodes in the payment path from $u$ to $v$.
Let $d(u,v)$ be the length of the shortest path from $u$ to $v$ in the network and let us assume that $u$ pays $f_{avg}^T$ to every intermediary node in the path. 
The expected fees encountered by $u$ with a stream of $N_u$ output transactions is the sum of costs which increases proportionally with the distance between any two users:

$$\fees = N_u \cdot \sum_{v \in V; v \neq u} d(u,v) \cdot f_{avg}^T \cdot \transacts$$

We write $\mathbb{E}^{fees}_{u_S}$ when we want to explicitly say that the user $u$ already added edges from $S$ to the network.
We note that when two users $u$ and $v$ are not connected, then $d(u,v) = +\infty$.

\subsectioninline{Objective of the user}
Here, we  combine all the costs calculated above and compute the utility function of a newly joining node.
The expected utility of a user $u$ under a given strategy $S \subseteq \Omega$ is the profit gained from collecting the fees, minus the fees paid for sending out transactions, and minus the costs of the channels. Formally, 

$$\utility = \revenue - \fees - \channelscost$$

We assume the utility of a disconnected node (i.e.\ a node that is not connected to any other node in the network) is $-\infty$.

The objective of $u$ is to select a subset of users to connect to as well as the amount of funds to lock into these channels that maximises their expected utility subject to the budget constraints. Formally:

$$\max_{S \in \Omega} \mathcal{U}_{u_{S}} \text{ s.t. } \sum_{(v,l_u) \in S} [C + l_u] \leq B_u $$

\section{Optimisation algorithms}\label{sec:optimization}
Having defined the utility and objective for a new user $u$ in~\Cref{sec:utility}, we now propose several algorithms to optimise the objective in this section. 
We begin by establishing some properties of our objective function.
We first show that our utility function is submodular but not necessarily monotone and not necessarily non-negative.
Thus, we cannot apply standard algorithms to optimise it efficiently with guarantees on the approximation ratio. 
We thus propose a series of constraints on the actions of the new user $u$ and define a solution for the objective in each constrained setting. 
We then provide a corresponding optimisation algorithm for each setting that comes with guarantees on the approximation ratio. In the following, let $[k]$ denote $\{1, \dots, k\}$.

\subsection{Properties of the objective function}\label{properties}
Whenever we add a new edge, its estimated average transaction rate will depend on the current topology of the network and the capacities of the channels in the network.
We first show that the objective function is submodular. 
Let $S \subset \Omega$ be a strategy. 
Note that we allow the algorithm to add more than one channel with the same endpoint $v$ but different amounts of funds $l_i$ to the strategy set $S$.

\begin{theorem}
The expected utility function $\utility$ is submodular.
\end{theorem}

\begin{proof}
We split $\utility$ into three components that sum to $\utility$ and show that each component is submodular. Since the sum of submodular functions is submodular, the claim follows.

We first rewrite $\utility$ as 
\begin{equation}\label{eq:submodular}
    \utility = \revenue + \left(- \fees\right) + \left(- \channelscost\right).
\end{equation}
Consider the configurations with two strategies $u_{S_1}, u_{S_2}$ with $S_1 \subseteq S_2$, and consider a pair $X = (x,l_{ux}) \notin S_2$. Recall that a function $g$ is submodular if $g(u_{S_2 \cup \{X\}}) -g(u_{S_2}) \leq g(u_{S_1 \cup \{X\}}) - g(u_{S_1})$.

Now first observe that 
\begin{align*}
    \mathbb{E}^{rev}_{u_{S1 \cup \{x\}}}- \mathbb{E}^{rev}_{u_{S1}} &= \mathbb{E}^{rev}_{u_{\{X\}}} = \lambda_{x u} \cdot f_{avg} = \mathbb{E}^{rev}_{u_{S2 \cup \{X\}}} - \mathbb{E}^{rev}_{u_{S2}}
\end{align*}

Hence the expected revenue function $\mathbb{E}^{rev}_{u_{S}}$ is submodular.
Note that in the calculations we assume that $\lambda_{xy}$ is a fixed value.

Now we show that the second component of \ref{eq:submodular} is submodular.
That is,
$-\mathbb{E}^{fees}_{u_{S}} = -\lambda_u \sum_{v \in V; v \neq u} d(u,v) \cdot f_{avg}^T \cdot \transacts$ is submodular. 
Let us denote the marginal contribution in terms of the expected fees of adding $X$ to strategy $S$ as $MC_S(X) :=  \mathbb{E}^{fees}_{u_{S}} -\mathbb{E}^{fees}_{u_{S \cup \{X\}}}$.
We note that $MC_S(X)$ only changes when one adds a pair $X=(x,l_{ux})$ to $S$, such that a shortest path from $u$ to some $v$ goes through the vertex $x$ in the new configuration $S \cup \{X\}$, i.e.:
$$MC_S(X) = 
\lambda_u f_{avg}^T \sum_{\substack{v \in V; v \neq u;\\x \in sp_{S \cup \{X\}}(u,v)}}  \transacts \Big[ d_{S}(u,v) - d_{S \cup \{X\}}(u,v)  \Big]$$
Recall that $d(u,v)$ as defined for two disconnected nodes $u,v$ is $+\infty$. Thus, $d_{S_1 \cup \{X\}}(u,v) - d_{S_1}(u,v) \leq 0$ as $X \notin S_1, S_2$. 
Moreover, as $v \in S_1, S_2$ are direct neighbours of $u$ in all configurations, then $|d_{S_1}(u,v)-d_{S_1 \cup \{X\}}(u,v)| > |d_{S_2}(u,v)-d_{S_2 \cup \{X\}}(u,v)|$. 
Hence, we conclude that $MC_{S_1}(X) > MC_{S_2}(X)$.
Note that in the calculations we assume that $\transacts$ is a fixed value.

Finally, we show that the last component $- \channelscost$ in \ref{eq:submodular} is submodular.
The marginal contribution of $X=(x,l_{ux})$ to the channel costs given $u_{S_1}$ is simply the cost of a single bidirectional channel between $u$ and $x$, i.e. $L_u(v,x)$. 
This is exactly equal to the marginal contribution given $u_{S_2}$. 
\end{proof}
Now, we show that although the objective function is submodular, it is unfortunately non-monotone. That is, for any two strategy sets $S_1, S_2$ with $S_1 \subset S_2$, it is not necessarily the case that $\mathcal{U}_{u_{s_1}} \leq \mathcal{U}_{u_{s_2}}$.

\begin{theorem}
The expected utility function $\utility$ is not necessarily monotone, but the modified utility function $\utility' = \revenue-\fees$ is monotonically increasing.
\end{theorem}

\begin{proof}
We analyse each component of $\utility$ separately. 
First, we note that a direct application of \cite{ersoy2019profit} shows that $\revenue$ is monotone increasing. Next, we look at expected fees:
$$-\mathbb{E}[\text{fees encountered by } u_{S}] = -\lambda_u \sum_{v \in V; v \neq u} d(u,v) \cdot f_{avg}^T \cdot \transacts.$$
The monotonicity of this function directly follows from the fact that for any $S_1 \subseteq S_2$, $d_{S_1}(u,v) \geq d_{S_2}(u,v)$. Thus, the function is monotonically increasing. Note that in the calculations we assume that $\transacts$ is a fixed value.

Finally, $- \channelscost$ is clearly a monotonically decreasing function. Since two components of $\utility$ are monotonically increasing and one component is monotonically decreasing, $\utility$ is non-monotone. 
\end{proof}
The final property we show about our objective function is that it is not necessarily non-negative.
\begin{theorem}
The expected utility function $\utility$ is not necessarily non-negative.
\end{theorem}
\begin{proof}
This follows from the observation that the sum of the cost of creating channels and the expected fees $\channelscost+\fees$ might easily get bigger than the expected revenue $\revenue$ when choosing some strategy $S \subseteq \Omega$. 
\end{proof}
\subsection{Fixed amounts of funds per channel}\label{sec:greedy}
We first show that if we restrict the amount of funds (say $l_1$) that the new user $u$ can lock in each created channel, we can achieve an approximation ratio of $1-\frac{1}{e}$. This setting is useful for users who want to minimize their computational cost.
The algorithm (described in Algorithm~\ref{alg:greedy}) that achieves this ratio in this setting is simple -- we greedily pick the $k$ best channels to connect with that maximize the expected revenue minus the expected fees.
Formally, let us define a simplified utility function $\utility'$ which is the sum of the expected revenue and the expected fees: $\utility' = \revenue + (- \fees)$.

We note that the simplified utility function $\utility'$ is submodular and monotone, as shown in~\ref{properties}.
Let us denote the maximum number of channels that can be created given $u$'s budget $B_u$ by $\maxsize := \lfloor\frac{B_u}{C+l_1}\rfloor$.
We can now maximize $\utility'$ and find the optimal set of vertices to connect to for each possible subset of vertices of size $k$. 
We do this for $k \in \{1,2, \ldots, \maxsize \}$, and then compare the results for all $k$.
Since the channel creation cost is now fixed for any choice of $k$ new channels, the $(1-\frac{1}{e})$-approximation we achieve when we greedily maximize $\utility'$ simply follows from the result in~\cite{greedyalg} since $\utility'$ is submodular and monotone. 

The next theorem shows that Algorithm~\ref{alg:greedy} returns a $(1 - \frac{1}{e})$-approximation and runs in time linear in \maxsize.
\begin{theorem}
Algorithm~\ref{alg:greedy} with inputs $\Omega = \{(v,l_1) \in V: v \neq u\}$ and $\maxsize$ returns a $(1 - \frac{1}{e})$-approximation of the optimum of $\utility'$. The result is computed in at most $\mathcal{O}(\maxsize \cdot n)$ number of estimations of the $\lambda_{uv}$ parameter.
\end{theorem}
\begin{proof}
To see that Algorithm~\ref{alg:greedy} returns a $(1 - \frac{1}{e})$-approximation of the optimum of $\utility'$, we need to see that in the algorithm for each possible $k$ we compute a $(1 - \frac{1}{e})$-approximation of $\mathcal{U'}$ (in $\mathcal{O}(n)$ time), because the function $\mathcal{U'}$ is submodular and monotonically increasing, then the overall solution that compares partial results gives a $(1-\frac{1}{e})$-approximation ratio for a fixed $k$. This in turn gives a $(1-\frac{1}{e})$-approximation ratio for each $k \in \{1,2, \ldots, \maxsize \}$.
\end{proof}

\begin{algorithm}[t!]
\SetAlgoLined
\DontPrintSemicolon
\caption{Greedy algorithm}\label{alg:greedy}
\KwData{$\Omega,\maxsize$}
 $P_S \gets \text{array indexed } 1, \ldots, \maxsize \text{ initialized with } P[i] = \emptyset$\;
 $P_U \gets \text{array indexed } 1, \ldots, \maxsize \text{ initialized with } P[i] = -\infty$\;
 $S \gets \emptyset$\;
 $A \gets \Omega$\;
\While{$|S| \leq  \maxsize$} {
    $X \gets argmax_{X\in A} [\mathcal{U'}_{u_{S \cup \{X\}}} - \mathcal{U'}_{u_{S}}]$\;
    $S \gets S \cup \{X\}$ \;
    $P_S[|S|] \gets S$\;
    $P_U[|S|] \gets \mathcal{U'}_{u_{S}}$\;
    $A \gets A \setminus \{X\}$\;
    }
$i \gets argmax_{i \in \{1, \ldots, M\}}[P_U[i]]$\;

\Return $P_S[i]$
\end{algorithm}

\subsection{Varying amount of funds per channel, discrete version}
Next, we give the new user a choice of locking varying amounts of capital in each channel. 
Enabling varying capital on channels depicts more accurately the realistic model of transaction distribution we leverage.
However, in order to achieve the same approximation ratio of $1-\frac{1}{e}$ as in the previous setting, we have to discretize the capital that can be locked into a channel to some minimal amount $m >0$.
That is, opening a channel would require injecting funds of the form $km$ for some $k \in \mathbb{N}$.
We impose this discretization constraint in order to perform an exhaustive search over all possible assignments of the budget $B_u$ to the capital in each channel.

We again operate on the modified utility function $\utility'$
and present an algorithm (described in Algorithm~\ref{alg:exhausitve}) that achieves the same approximation ratio of $1-\frac{1}{e}$.
In more detail, given a parameter $m$, Algorithm~\ref{alg:exhausitve} firstly divides the budget $B_u$ to $\frac{B_u}{m} $ units that can be spent. 
Then, the algorithm divides these units into $k+1$ parts (where $k = \lfloor \frac{B_u}{C} \rfloor $ is a bound on the number of channels that $u$ can possibly create).
Finally, for each possible division, it runs Algorithm~\ref{alg:greedy} (again by temporarily skipping the channel costs) in each step locking the capital assigned to this channel in the division. 
Let us denote $T := \binom{\frac{B_u}{m}}{\frac{B_u}{C}+1}$. 
\begin{algorithm}[t!]
\SetAlgoLined
\DontPrintSemicolon
\caption{Exhaustive search over channel funds}\label{alg:exhausitve}
\KwData{$V, B_u, m$}
 $k = \lfloor \frac{B_u}{C} \rfloor$\;
 $D = \text{array of all divisions of }[\lfloor\frac{B_u}{m}\rfloor] \text{ to } k+1 \text{ parts}$\;
 $D_S \gets \text{array indexed } 1, \ldots, |D| \text{ initialized with } D_S[i] = \emptyset$\;
\For{$i \in [|D|]$} {
$(l_1, \ldots, l_{k+1}) \gets D[i]$\; 
    $D_S[i] \gets $ the output of Algorithm~\ref{alg:greedy} run on $M = k$ with a restriction that in every step $j$ of \emph{while} loop in the algorithm a channel of capacity $l_j$ is selected\;
    }
$i \gets argmax_{i \in \{1, \ldots, |D|\}}\mathcal{U'}_{u_{D_S[i]}}$\;
\Return $D_S[i]$
\end{algorithm}

\begin{theorem}
Algorithm~\ref{alg:exhausitve} with inputs $V$, budget $B_u$, and parameter $m$ returns a $(1 - \frac{1}{e})$-approximation of the optimum of $\utility'$. The result is computed in at most $\mathcal{O}(T \cdot \frac{B_u}{C} \cdot n)$ steps.
\end{theorem}
\begin{proof}
The budget $\frac{B_u}{m}$ can be split to at most $k = \lfloor \frac{B_u}{C} \rfloor$ parts in at most $\binom{\frac{B_u}{m}}{k+1}$ cases. 
Algorithm~\ref{alg:greedy} is run as a subroutine of Algorithm~\ref{alg:exhausitve}. The main routine iterates through all possible combinations of amounts locked to channels, each of them giving the $(1-\frac{1}{e})$-approximation for the selected assignments of funds.
\end{proof}

We note that there is a trade-off between the choice of $m$ and the run time of Algorithm~\ref{alg:exhausitve}: a larger $m$ would reduce the search space and hence the runtime of the algorithm. However, it would reduce the control over the capital the user could lock into any particular channel.

\subsection{Varying amount of funds per channel, continuous version}
In this section, we remove the previous discrete constraint on the capital the new user $u$ can inject into the channel, that is, $u$ can now inject funds of the form $m \in \mathbb{R}+$ into any channel.
We sketch a polynomial-time $\frac{1}{5}$-approximation algorithm for the optimisation problem:
let us first denote the total expected on-chain transaction cost for a user $u$ with an average output stream of $N_u$ transactions as $C_u := \frac{N_u\cdot C}{2}$. 
That is, $C_u$ represents the total expected cost for user $u$ when $u$ transacts entirely on the blockchain.
One can now consider what we term the benefit function, which is simply the sum of $C_u$ and the utility of $u$ when $u$ joins the network with strategy $S$. Formally, we denote this function by $\utility^{b} := C_u + \utility$.
Intuitively, the benefit function captures the potential benefit $u$ would gain from transacting with other users over the PCN rather than on the blockchain. 

We observe that $\utility^{b}$ will stay submodular and positive whenever the user chooses channels $(u,v)$, such that $$\fees + \frac{B_u}{C} \cdot L_u(v,l) < C_{u}.$$ 
As such, we can apply the algorithm and result of Lee et al.~\cite{LeeMNS09} for optimising submodular and non-negative functions to $\utility^{b}$ to achieve a $\frac{1}{5}$-approximation of $\utility^{b}$.

\section{Structural properties of simple graph topologies}
\label{sec:nash}
In this section, we complement our study of optimisation algorithms for users in the payment channel network (\Cref{sec:optimization}) with a study of structural properties of simple graph topologies given the transaction model between users as defined in Section~\ref{subsection:model_params}.
We are particularly interested in properties of stable networks, that is, networks that are in a Nash Equilibrium where no user can increase their utility by any unilateral change in strategy.
Stability is an important notion in the context of PCNs as this has implications not only on the choice of which nodes to connect to for a new user~\cite{avarikioti2020ride} but also on payment routing and finding off-chain rebalancing cycles for existing users to replenish depleted channels~\cite{AvarikiotiPSSTY22}.
We are also interested in the parameter space of our model under which specific graph topologies form a Nash Equilibrium.  

We use the following assumptions and notations in our analysis in this section:
\begin{enumerate}
    \item Recall from Equations~\ref{eq:pe}~and~\ref{eq:rev} in~\Cref{sec:preliminaries} that the expected revenue of a user $u$ can be written as:
    $
        \revenue = \sum_{v_i \in Ne(u)} \lambda_{ u v_i} \cdot f_{avg} =
        \sum_{\substack{ v_1 \neq v_2\\ v_1,v_2 \in V \setminus \{u\}}} \frac{m_u(v_1,v_2)}{m(v_1,v_2)} \cdot N_{v_1} \cdot p^{trans}_{v_1,v_2} \cdot f_{avg}$
    
    We denote $b := N_{v_1} \cdot f_{avg}$ and assume it is constant for $v_1 \in V \setminus \{u\}$.
    \item We denote $a := N_u \cdot f_{avg}^T$.
    \item For any $s>0$ and $n \in \mathbb{N^*}$, we denote $H_n^s := \sum_{k=1}^n \frac{1}{k^s}$.
    \item All the players create channels of equal cost $l$.
\end{enumerate}

\subsection{Upper bound on the longest shortest path containing a hub}
An interesting question is how large is the diameter of stable networks with highly connected nodes. 
In the context of PCNs, this has implications on efficient payment routing algorithms\cite{RoosMKG18,SivaramanVRNYMF20,PietrzakS0Y21,landmark}.
As a first step to answering this question, we derive an upper bound on the longest shortest path in a stable network that contains a hub node, i.e., an extremely well-connected node that transacts a lot with other nodes in the network.
Let us select a hub node $h$ and consider the longest shortest path that $h$ lies on (if there are multiple we simply select one of them arbitrarily).
We denote the length of the path by $d$.
The following theorem derives an upper bound on $d$ for a stable network.

\begin{theorem}
$d$ is upper bounded by $2(\frac{\frac{C+\epsilon}{2} - \lambda_e \cdot f}{p_{\min} \cdot N \cdot f}) +1$.
\end{theorem}
\begin{proof}
Let $P = (v_0, v_1, \dots, v_d)$ be the path. 
Consider the addition of an edge $e$ between $v_{\lfloor \frac{d}{2} \rfloor -1}$ and $v_{\lfloor \frac{d}{2} \rfloor +1}$. 
Denote by $\lambda_e$ the minimum rate of transactions going through the edge $e$ in both directions, i.e. 
$\lambda_e := \min\{\lambda_{(v_{\lfloor \frac{d}{2} \rfloor -1}, v_{\lfloor \frac{d}{2} \rfloor +1})}, \lambda_{(v_{\lfloor \frac{d}{2} \rfloor +1}, v_{\lfloor \frac{d}{2} \rfloor -1})}\}$.

Now consider the set of directed shortest paths $S$ such that each path $s_i \in S$ is a sub sequence of $P$ and one end point of $s_i$ lies in $\{v_0, \dots, v_{\lfloor \frac{d}{2} \rfloor -1}\}$ and the other end point of $s_i$ lies in $\{v_{\lfloor \frac{d}{2} \rfloor +1}, \dots, v_d\}$. Let $p_{i}$ be the probability that $s_i$ is selected, with probabilities of directed paths being selected as defined by the probability of the source of the path transacting with the sink (refer to \ref{eq:prob} for more details). Let $p_{\min} := \min_i p_i$.

We know the cost (split equally) of creating the edge $e$ is at least $\frac{C+\epsilon}{2}$. Since the network is stable, this implies that the cost of creating $e$ is larger than any benefits gained by the $2$ users $v_{\lfloor \frac{d}{2} \rfloor -1}$ and $v_{\lfloor \frac{d}{2} \rfloor +1}$ by creating $e$. That is,

\begin{equation}\label{eq:length}
    \frac{C+\epsilon}{2} \ge \lambda_e \cdot f + N \cdot p_{\min} \cdot f \cdot \lfloor\frac{d}{2} \rfloor,
\end{equation}
\noindent where the first term on the RHS of the inequality is the minimum (among the two parties $v_{\lfloor \frac{d}{2} \rfloor -1}$ and $v_{\lfloor \frac{d}{2} \rfloor +1}$) of the average revenue gained by adding the edge $e$. 
The second term on the RHS of the inequality is a lower bound on the average amount of fees saved by $v_{\lfloor \frac{d}{2} \rfloor -1}$ and $v_{\lfloor \frac{d}{2} \rfloor +1}$.
Rearranging, this implies that $d \le 2(\dfrac{\dfrac{C+\epsilon}{2} - \lambda_e \cdot f}{p_{\min} \cdot N \cdot f}) +1.$
\end{proof}

Note that since a hub node is on the path, as long as it is not directly in the middle of the path (i.e. vertex $v_{\lfloor \frac{d}{2} \rfloor}$), $p_{\min}$ should be fairly large as hubs are typically high degree vertices. Moreover,
if a hub node is on a diametral path, we extract a meaningful bound on the diameter of a stable network.

\subsection{Stability of simple graph topologies}
In this section, we study some simple graph topologies, and the parameter spaces of the underlying transaction distribution under which they form a Nash Equilibrium.
We restrict our analysis to these simple topologies because computing Nash Equilibria for a general graph using best response dynamics is NP-hard (see Theorem $2$ in~\cite{avarikioti2020ride}).
As mentioned in Section~\ref{subsec:transactions}, we assume the underlying transaction distribution that gives the probability of any two nodes transacting with each other is the Zipf distribution.

We firstly show that when the scale parameter $s$ of the Zipf distribution is large (i.e. the distribution is heavily biased towards transacting with only high-degree nodes), the star graph is a Nash Equilibrium.

\begin{theorem}
The star graph with the number of leaves $\geq 4$ is a Nash Equilibrium when nodes transact with each other according to the Zipf distribution with parameter $s$ such that $\frac{1}{2^s}$ is negligible, i.e. $\frac{1}{2^s} \approx 0$.
\end{theorem}
\begin{proof}

First note that, because $\frac{1}{2^s}$ is negligible, then all leaf nodes have negligible expected revenue. 
Now consider a leaf node $u$. The costs of the leaf node $u$ are triggered by transacting with the central node. If $u$ removes the edge between $u$ and the central node, and replaces this connection with a set of edges to other leaf nodes, $\fees$ can only rise, as the central node still remains the one with the highest degree. 

The central node may want to delete all of its edges, but this will result only in lowering its $\revenue$. The $\fees$ may not go down, because the central node already communicates directly with all leaf nodes. 
\end{proof}

Secondly, we establish the necessary conditions that make the star graph a Nash Equilibrium in general.

\begin{restatable}[]{theorem}{starnash}
\label{thm:star_nash}
The star graph with the number of leaves $n \geq 2$ is a Nash Equilibrium when nodes transact with each other according to the Zipf distribution with parameter $s\geq 0$ whenever the following conditions hold:
\begin{enumerate}
    \item $a/H_n^s \leq 2^s \cdot l \cdot 1$,
    \item $b \cdot \frac{i}{2} \cdot \frac{H_{i+1}^s-1-1/2^s}{H_n^s} + a \cdot \frac{H_{i+1}^s-1}{H_n^s} \leq l \cdot (i)$ (for $2 \leq i \leq n-1$),
    \item $b \cdot \frac{i}{2} \cdot \frac{H_{n}^s-1-1/2^s}{H_n^s} + a \frac{H_{i+1}^s-2}{H_n} \leq l \cdot (i-1)$ (for $2 \leq i \leq n-1$).
\end{enumerate}
\end{restatable}
\begin{proof}

    Firstly, we prove that the \emph{central node} is in Nash Equilibrium in the star graph. Since the central node is connected to all other nodes, adding an additional channel to any node just increases the channel creation cost and thus decreases the utility for the central node. Removing a single edge disconnects the central node from a user and thus leads to infinite cost. Thus the central node has no incentive to switch to a different strategy. 
    
    Secondly, we prove when any \emph{leaf node} is also in Nash Equilibrium in the star graph.
    For every strategy defined below, we calculate expected revenue $\revenue$, expected costs $\fees$, and channels cost $L$ of the node $u$ after changes.

    \begin{environment}\leavevmode
    -- By default a leaf node $u$ will not add/remove any edges. \begin{itemize}
            \item $\revenue = 0$ 
            \item She interacts with 1 central node with $rf = 1$, and $n-1$ leaf nodes with\\ $rf = \frac{H_n^s-1}{n-1}$. $\sum_{v' \in V \setminus \{u\} = 1 + (n-1) \cdot \frac{H_n^s-1}{n-1}} = H_n^s$, \\
            $\fees = - a (n-1) \frac{\frac{H_n^s-1}{n-1}}{H_n^s} = -a \cdot \frac{H_{n}^s-1}{H_n^s}$, $L = - l \cdot 1$.
        \end{itemize}
    -- A leaf node may also try to add connections to $n-1$ other leaf nodes.
        \begin{itemize}
            \item the other leaf nodes $v'$ interact directly, with $2$ nodes (the central node, and the nodes that changes its strategy, both connected to $n-1$ other nodes) with $rf = \frac{1 + 1/2^s}{2}$, and indirectly with $n-2$ other nodes with $rf = \frac{H_n^s-1-1/2^s}{n-2}$. $\sum_{v'' \in V \setminus \{v'\}} rf(v'') = H_n^s$ \\
            $\revenue = b \cdot [2 \cdot 1/2]\binom{n-1}{2}\frac{\frac{H_n^s-1-1/2^s}{n-2}}{ H_n^s } = b \cdot \frac{n-1}{2} \cdot \frac{H_n^s-1 - 1/2^s}{H_n^s}$
            \item $\fees = 0$, $L = - l \cdot n $.
        \end{itemize}
    -- The leaf node may also add connections to $n-1$ leaf nodes and remove the connection with the central node.
        \begin{itemize}
            \item the other leaf nodes $v'$ interact directly, with $2$ nodes (the central node, and the nodes that changes its strategy, both connected to $n-2$ other nodes) with $rf = \frac{1 + 1/2^s}{2}$,  and indirectly with $n-2$ other nodes with $rf = \frac{H_n-1-1/2^s}{n-2}$. $\sum_{v'' \in V \setminus \{v'\}} rf(v'') = H_n^s$ \\
            $\revenue = b\cdot[2\cdot 1/2]\binom{n-1}{2}\frac{\frac{H_n-1-1/2^s}{n-2}}{ H_n^s} = b \cdot \frac{n-1}{2} \cdot \frac{H_n^s-1-1/2^s}{H_n^s}$
            \item $\fees = -a/H_n^s$, $L = - l \cdot (n -1)$.
        \end{itemize}    
     -- The leaf node can add connection to only one other $1$ leaf node.
        \begin{itemize}
            \item $\revenue = 0$
            \item $u$ connects to one central node $rf = 1$, $1$ node with $rf = \frac{1}{2^s}$, and $n-2$ other nodes - $rf = \frac{H_{n}^s-1-\frac{1}{2^s}}{n-2}$. $\sum = H_n^s$
            
            $\fees = - a (n-2) \cdot \frac{\frac{H_{n}^s-1-\frac{1}{2^s}}{n-2}}{H_n^s} = -a \cdot (H_{n}^s-1-\frac{1}{2^s})/H_n^s$,  $L = - l \cdot 2 $.
        \end{itemize}
    -- The leaf node can add connections to $2 \leq i \leq n-2$ other leaf nodes.
        \begin{itemize}
            \item the leaf nodes $v'$ that $u$ connects to interact with $1$ central node ($rf = 1$), the $u$ node ($rf = 1/2$), $i-1$ other nodes that $u$ connects to $rf = \frac{H_{i+1}^s - 1 - 1/2}{i-1}$, and $n - H_{i+1}^s$ other nodes.
            
            $\sum_{v'' \in V \setminus \{v'\}} rf(v'') = H_n^s$ \\
            $\revenue = b\cdot[2\cdot 1/2]\binom{i}{2}\frac{\frac{H_{i+1}^s-1-1/2}{i-1}}{ H_n^s } = b \cdot \frac{i}{2} \cdot \frac{H_{i+1}^s-1-1/2^s}{H_n^s}$
            \item From the perspective o f $u$, the central node has $rf = 1$, the nodes that the $u$ connects to have $rf = \frac{H_{i+1}^s-1}{i}$, the other nodes have $rf = \frac{H_n^s-H_{i+1}^s}{n-i-1}$. $\sum = H_n^s$.
            
            $\fees = - a (n-i-1) \frac{H_n^s - H_{i+1}^s}{n-i-1}/H_n^s = -a \cdot (H_n^s - H_{i+1}^s)/H_n^s$, $L = - l \cdot (i+1) $.
        \end{itemize}
    -- The leaf node can add connections to $2 \leq i \leq n-2$ leaf nodes and remove the connection with the central node.
        \begin{itemize}
            \item the other leaf nodes $v'$ interact directly, with $2$ nodes, the central node with $rf = 1$, the $u$ node with $rf=1/2$ and indirectly with $i-1$ other nodes with $rf = \frac{H_{i+1}-1-1/2}{i-1}$. $\sum_{v'' \in V \setminus \{v'\}} rf(v'') = H_n$ \\
            $\revenue = b\cdot[2\cdot 1/2]\binom{i}{2}\frac{\frac{H_{i+1}^s-1-1/2^s}{i-1}}{ H_{i+1}^s} = b \cdot \frac{i}{2} \cdot \frac{H_{i+1}^s-1-1/2^s}{H_n}$
            \item From the perspective o f $u$, the central node has $rf = 1$, the nodes that the $u$ connects to have $rf = \frac{H_{i+1}^s-1}{i}$, the other nodes have $rf = \frac{H_n^s-H_{i+1}^s}{n-i-1}$. $\sum = H_n^s$.
            
            $\fees = - a \cdot [ (n-i-1) \frac{H_n^s - H_{i+1}^s}{n-i-1}/H_n^s + 1/H_n^s ]= -a \cdot (H_n^s - H_{i+1}^s+1)/H_n^s$, $L = - l \cdot i $. \end{itemize}
    Now we compare the utility gained by switching to each strategy as opposed to sticking to the default strategy:
     \paragraph{(1) vs (2).} If (1) remains a NE then:
    \begin{gather*}
        -a \cdot \frac{H_n^s-1}{H_n^s} - l \cdot 1 \geq b \cdot \frac{n-1}{2} \cdot \frac{H_n^s-1-1/2^s}{H_n^s} - l \cdot (n) \\
        \iff\\
        a \cdot \frac{H_n^s-1}{H_n^s} + b \cdot \frac{n-1}{2} \cdot \frac{H_n^s-1-1/2^s}{H_n^s} \leq l \cdot (n-1)
    \end{gather*}
    
      \paragraph{(1) vs (3).} If (1) remains a NE, then for any value of the parameter $s \geq 0$:
    
    \begin{gather*}
         -a \cdot \frac{H_n^s-1}{H_n^s} - l \cdot 1 \geq b \cdot \frac{n-1}{2} \cdot \frac{H_n^s-3/2}{H_n^s} - l \cdot (n-1) - \\
        - a/H_n^s \iff\\
        a \cdot \frac{H_n^s-2}{H_n^s} + b \cdot \frac{n-1}{2} \cdot \frac{H_n^s-1-1/2^s}{H_n^s} \leq  l \cdot (n-2)
    \end{gather*}
    
     \paragraph{(1) vs (4).} If (1) remains a NE:
    
    \begin{gather*} -a \cdot \frac{H_n^s-1}{H_n^s} - l \cdot 1 \geq - a \cdot (H_n^s - 1 -1/2^s)/H_n^s - \text{cost of 2}  \\
        \iff\\
        a/H_n^s \leq 2^s \cdot l \cdot 1
    \end{gather*}

     \paragraph{(1) vs (5).} If  (1) remains a NE:
    \begin{gather*}  -a \cdot \frac{H_n^s-1}{H_n^s} - l \cdot 1 \geq b \cdot \frac{i}{2} \cdot \frac{H_{i+1}^s-1-1/2^s}{H_n^s} -\frac{a \cdot (H_n^s - H_{i+1}^s)}{H_n^s} - \\- l \cdot (i + 1) 
        \iff\\
         b \cdot \frac{i}{2} \cdot \frac{H_{i+1}^s-1-1/2^s}{H_n^s} + a \cdot \frac{H_{i+1}^s-1}{H_n^s} \leq l \cdot i \\
    \end{gather*}
     \paragraph{(1) vs (6).} If (1) remains a NE:
    
    \begin{gather*} -a \cdot \frac{H_n^s-1}{H_n^s} - l \cdot 1 \geq b \cdot \frac{i}{2} \cdot \frac{H_{n}^s-1-1/2^s}{H_n^s} -a \cdot \frac{1+H_n^s - H_{i+1}^s}{H_n^s}-\\ - l \cdot i 
        \iff\\
        b \cdot \frac{i}{2} \cdot \frac{H_{n}^s-1-1/2^s}{H_n^s} + a \frac{H_{i+1}^s-2}{H_n} \leq l \cdot (i-1)\\
    \end{gather*}
    \end{environment}
    \vspace{-1.4cm}
\end{proof}

{
\vspace{11pt} Given the result above, we show that if the scale parameter of the distribution is \emph{only} moderately large ($s \geq 2$) and not too many messages are sent out in the network (i.e. $a/H_n^s,b/H_n^s \leq l $), then the star graph is still a Nash Equilibrium. The values $a/H_n^s,b/H_n^s \leq l \cdot 1$ give a bound on the transactions sent to the highest ranked node of a user.
}
\begin{theorem}
The star graph with a number of leaves $n \geq 2$ is a Nash Equilibrium when nodes follow the Zipf distribution with parameter $s\geq 2$ whenever the cost of all edges is equal, and $a/H_n^s,b/H_n^s \leq l \cdot 1$. 
\end{theorem}

\begin{proof}
Taking the conditions from Theorem~\ref{thm:star_nash}: 

\begin{enumerate}
    \item $b \cdot \frac{i}{2} \cdot \frac{H_{i+1}^s-1-1/2^s}{H_n^s} + a \cdot \frac{H_{i+1}^s-1}{H_n^s} \leq l \cdot (i)$ (for $2 \leq i \leq n-1$),
    \item $b \cdot \frac{i}{2} \cdot \frac{H_{n}^s-1-1/2^s}{H_n^s} + a \frac{H_{i+1}^s-2}{H_n^s} \leq l \cdot (i-1)$ (for $2 \leq i \leq n-1$),
    \item $a/H_n^s \leq 2^s \cdot l \cdot 1$.
\end{enumerate}

We can see that with our assumptions, condition $3$ holds trivially as $a/H_n^s \leq l \cdot 1$. Moreover, whenever the cost of all edges is equal, conditions $1,2$ are more restrictive, whenever $i$ increases, so the most restrictive case is when $i  = n-1$. Now, because $a/H_n^s \leq l \cdot 1$, condition $2$ is more restrictive than $1$. Finally, condition $2$ holds, because $a/H_n^s,b/H_n^s \leq l \cdot 1$, and for $s \geq 2$,
$H_n^s = \sum_{i = 1}^{n} \dfrac{1}{i^s} \leq \sum_{j = 1}^{+\infty} j \cdot \dfrac{1}{j^s} \leq 2.$
\end{proof}
We also show that the path graph essentially will never become a Nash Equilibrium.
\begin{theorem}
A path graph is never a Nash Equilibrium when nodes transact with each other according to the Zipf distribution with parameter $s\geq 0$.
\end{theorem}
\begin{proof}
Since the cost of any edge is split equally between both parties, the endpoints of the path would always prefer joining to a node that is not an endpoint of the path. In this case, even when $s=0$, their expected revenue factor still remains $0$, but the cost of the expected fees naturally gets lower. 
\end{proof}
We finally show that circle graph cannot be a NE when it is sufficiently large. 
\begin{restatable}{theorem}{circle}\label{thm:circle}
The Circle graph does not form a Nash Equilibrium for all $n \ge n_0$, for some $n_0$, when nodes transact with each other according to the Zipf distribution with $s\geq 0$.
\end{restatable}
\begin{proof}
Assume that we have a circle graph with $n+1$ nodes.
-- Default strategy for a node $u$ is not to add or remove any edges. \begin{itemize}
            \item In this case $u$ is an intermediary node to all of the pairs of nodes for which the shortest path goes through this node. They rank each other equally, so each node ranks other nodes with equal $rf = H_n^s/n$, thus $\sum rf = H_n^s$, finally
            $\revenue = b\cdot \frac{H_n^s/n}{H_n^s} 2 \cdot (\binom{n}{2}-\binom{n/2}{2}-n/2\cdot n/2) \approx \frac{b}{n} \cdot n^2/4$ 
            \item The node $u$ interacts with $n$ nodes with $rf = H_n^s/n$. $2$ of them are in distance $0$, $2$ are in distance $1$, and so on. Finally at most $2$ of them are in distance $\lfloor n/2 \rfloor$. 
            $\fees = - a \cdot \frac{H_n^s/n}{H^s_n} \cdot 2 \cdot ( 1 + 2 + \ldots + n/2) \approx \frac{-a}{n} \cdot n^2/4$.
            \item $L = - l \cdot 1 $.
           \end{itemize}
-- A strictly better strategy for the node $u$ is to connect to its opposite node. \begin{itemize}
            \item In this case $u$ is an intermediary node to all of the pairs of nodes for which the shortest path goes through this node. The opposite node $u$ ranks $u$ with $rf = 1$ and all of the other nodes with $rf = \frac{H_n^s-1}{n-1}$ the other nodes rank 2 nodes with $rf = \frac{1+1/2^s}{2}$, and all of the other nodes with $rf = \frac{H_n^s-1-1/2^s}{n-2}$, thus $\sum rf = H_n^s$. We will thus asymptotically count only the weakest $rf = \frac{H_n^s-1-1/2^s}{n-2}$ factor. Finally
            $\revenue = b\cdot \frac{H_n^s-1-1/2^s}{n-2} \cdot 2\cdot(\frac{n}{4}\cdot\frac{n}{2}+ \frac{1}{2}\cdot\frac{n}{4}\
            \cdot \frac{n}{4}) \approx \frac{b}{n} \cdot n^2(5/16)$ 
            \item The node $u$ interacts with $n-1$ nodes with $rf = (H_n^s-1)/(n-1)$, and directly with one node with $rf=1$. We calculate the closeness as:
            \begin{align*}
            \nonumber
            \fees \leq - a \cdot \frac{(H_n^s-1)/(n-1)}{H^s_n} \cdot \frac{3 \frac{n}{4}(\frac{n}{4}-1)}{2} \\+
            \frac{n/2+n/4}{2}\cdot\frac{n}{4}) =  \frac{-a(H_n^s-1)/(n-1)}{H^s_n} \leq \frac{3}{16}n^2.    
            \nonumber
            \end{align*}
            \item $L = - l \cdot 1.$
           \end{itemize}
           \vspace{-0.4cm}
\end{proof}
\vspace{3pt}

\section{Related work}
Strategic aspects of cryptocurrencies, and more generally the blockchain technologies, 
have attracted a lot of attention in the literature \cite{chen2021game,amoussou2020rational,lewenberg2015bitcoin} as by their very nature, they are created to facilitate interactions between self-interested parties in a decentralised manner.

Apart from the works discussed in the introduction (\cite{avarikioti2020ride,avarikioti2019payment,ersoy2019profit,lnecon}), perhaps the closest research line to which our paper contributes is the one on creation games. 
In a well-known work by Fabrikant et al.~\cite{FabrikantLMPS03}, players choose a subset of other players to connect to in order to minimise their total distance to all others in the network. 
The result of Fabrikant et al. was later strengthened by Albers et al.~\cite{AlbersEEMR14}, and also extended to the weighted network creation game setting. 
Ehsani et al.~\cite{Ehsani2011} considers the network creation game with a fixed budget for each player, thus constraining the number of connections each player can make. Another well-known body of research of this kind are network formation games~\cite{bala2000noncooperative,jackson2005survey}.
All of these works, however, consider the problem of network creation in general networks which do not take into account fees and channel collateral which are specific to  PCNs.



Our work is also closely related to the study of stable network topologies for real-world networks (e.g. social and communication networks) that are formed by the interaction of rational agents~\cite{DemaineHMZ07,Bilo0LLM21}.
Demaine et al.~\cite{DemaineHMZ07} show that all equilibrium networks satisfy the small world property, that is, these networks have small diameters.
Bilo et al.~\cite{Bilo0LLM21} establish properties on the diameter, clustering and degree distribution for equilibrium networks. 
In~\cite{avarikioti2019payment,avarikioti2020ride}, Avarikioti et al. consider stable graph topologies in the context of PCNs.
Our work extends the analysis of Avarikioti et al.~\cite{avarikioti2020ride} and considers stable graph topologies in PCNs under a non-uniform distribution of transactions between users. 


\section{Conclusion and Future Work}\label{sec:conclusion}
In this paper, we modeled and analysed the incentive structure behind the creation of PCNs. We first focused on the perspective of a new user who wants to join the network in an optimal way. To this end, we defined a new user’s utility function in terms of expected revenue, expected fees, on-chain cost of creating channels, and opportunity costs, while accounting for realistic transaction distributions.

We also introduced a series of approximation algorithms under specific constraints on the capital distribution during the channel creation: 
(a) We first presented a linear time $1-\frac{1}{e}$ approximation algorithm when a user locks a fixed amount to all channels; thus, providing an efficient approach for users who wish to lower computational costs.
(b) We further provided a pseudo-polynomial time $1-\frac{1}{e}$ approximation algorithm  when users may lock varying, but discretized by $m$, amounts to different channels. This setting applies to most real-life scenarios  but comes with a computational overhead that depends on $m$.
(c) Finally, we proposed a $1/5$ approximation solution when a user can pick the amounts from a continuous set. 
    We used a modified utility function, the benefit function, which may be leveraged by a user to test whether assuming continuous funds yields  unexpected profits.
Altogether, our results in this section show that depending on the number of assumptions a new user joining a PCN wants to make, the user has a range of solutions to deploy to optimize the way they connect to the network. 

Lastly, we analysed the parameter spaces in our underlying model and conditions under which the star, path, and circle graph topologies form a Nash Equilibrium. Our analysis indicates that under a realistic transaction model, the star graph is the predominant topology, enhancing the results of~\cite{avarikioti2020ride}.

We highlight three interesting directions for future work. First, it would be beneficial to develop more advanced algorithms for maximizing the general utility function that also come with guarantees on the approximation ratio.
Second, we believe there are still avenues in which our model can be made more realistic, for instance, by considering a more realistic cost model that takes into account interest rates as in~\cite{lnecon}. 
Lastly, as the accuracy of our model depends on estimations of the underlying PCN parameters, for instance, the average total number of transactions and the average number of transactions sent out by each user, developing more accurate methods for estimating these parameters may be helpful.

 \section*{Acknowledgments}
    The work was partially supported by the Austrian Science Fund (FWF) through the project CoRaF (grant 2020388). It was also partially supported by NCN Grant 2019/35/B/ST6/04138 and  ERC Grant 885666.

\bibliographystyle{IEEEtran} 
\bibliography{main}

\end{document}